\documentclass[11pt]{article}
\usepackage{amsmath,amssymb,amsthm,amsfonts,latexsym,xspace,enumerate,graphicx,color,eucal,upgreek,mathtools}
\usepackage[backref,colorlinks,bookmarks=true]{hyperref}

\newtheorem{theorem}{Theorem}[section]
\newtheorem*{namedtheorem}{\theoremname}
\newcommand{\theoremname}{testing}
\newenvironment{named}[1]{ \renewcommand{\theoremname}{#1} \begin{namedtheorem}} {\end{namedtheorem}}

\newtheorem{claim}[theorem]{Claim}

\newtheorem{proposition}[theorem]{Proposition}

\theoremstyle{definition}
\newtheorem{definition}[theorem]{Definition}

\theoremstyle{plain}
\newtheorem{Alg}{Algorithm}

\renewenvironment{proof}{\noindent{\textbf{Proof:}}} {$\blacksquare$\vskip \belowdisplayskip}

\newcommand{\ignore}[1]{}

\renewcommand{\Pr}{\mathop{\bf Pr\/}}                    

\newcommand{\Ex}{\mathop{\bf E\/}}

\newcommand{\littlesum}{\mathop{{\textstyle \sum}}}

\newcommand{\poly}{\mathrm{poly}}

\newcommand{\R}{\mathbb R}

\newcommand{\Z}{\mathbb Z}

\newcommand{\eps}{\epsilon}

\newcommand{\zo}{\{0,1\}}
\newcommand{\co}{\colon}

\newcommand{\ul}[1]{\underline{#1}}

\newcommand{\calB}{{\mathcal{B}}}
\newcommand{\calC}{{\mathcal{C}}}

\newcommand{\calP}{\mathcal{P}}
\newcommand{\calR}{\mathcal{R}}
\newcommand{\calS}{\mathcal{S}}

\newcommand{\bone}[1]{{\boldsymbol{1}[#1]}}

\newcommand{\bV}{{\boldsymbol{V}}}

\newcommand{\itm}{solution\xspace}

\newcommand{\itms}{solutions\xspace}
\newcommand{\anitem}{a solution\xspace}

\newcommand{\allitems}{\calS}
\newcommand{\subitems}{\calR}
\newcommand{\candidates}{\calC}

\newcommand{\weight}{value\xspace}
\newcommand{\weights}{values\xspace}
\newcommand{\Weights}{Values\xspace}
\newcommand{\Wts}{V}
\newcommand{\bWts}{\bV}

\newcommand{\objs}{objectives\xspace}

\newcommand{\objvec}[2]{\mathrm{Obj}^{#2}(#1)}

\newcommand{\PO}{\mathbf{PO}}
\newcommand{\OK}{\mathbf{OK}}
\newcommand{\ok}{OK\xspace}

\newcommand{\T}{\mathbf{T}}

\newcommand{\Trans}{\mathtt{Witness}}
\newcommand{\transcript}{testimony\xspace}
\newcommand{\transcripts}{testimonies\xspace}
\newcommand{\Transcripts}{Testimonies\xspace}
\newcommand{\Transcript}{Testimony\xspace}

\newcommand{\Recon}{\ul{\mathtt{Reconstruct}}}

\newcommand{\Boxx}{Box\xspace}
\newcommand{\boxx}{box\xspace}
\newcommand{\Boxes}{Boxes\xspace}
\newcommand{\boxes}{boxes\xspace}
\newcommand{\tbox}[1]{$#1$-box\xspace}
\newcommand{\tboxes}[1]{$#1$-boxes\xspace}
\newcommand{\BoxL}{\calB}
\newcommand{\Bx}{B}
\renewcommand{\dim}{\mathrm{dim}}
\newcommand{\dimension}{dimension\xspace}

\newcommand{\pdfbd}{\upphi}

\newcommand{\J}{J}
\newcommand{\A}{A}
\newcommand{\free}[1]{M(#1)}
\newcommand{\Masking}{Masking\xspace}
\newcommand{\masking}{masking\xspace}

\newcommand{\WJ}{\free{\J} \circ \Wts}
\newcommand{\bWJ}{\free{\J} \circ \bWts}
\newcommand{\bWJbar}{\overline{\free{\J}} \circ \bWts}
\newcommand{\WJbar}{\overline{\free{\J}} \circ \Wts}

\newcommand{\Y}{Y}

\begin{document}

\title{Pareto Optimal Solutions for Smoothed Analysts}
\author{Ankur Moitra\thanks{Massachusetts Institute of Technology, \texttt{moitra@mit.edu}. Part of this research done while at Microsoft Research New England. Supported in part by a Fannie and John Hertz Foundation Fellowship.} \and Ryan O'Donnell\thanks{Carnegie Mellon University, \texttt{odonnell@cs.cmu.edu}. Part of this research done while visiting Microsoft Research New England; part done while at the Institute for Advanced Study. Supported by NSF grants CCF-0747250 and CCF-0915893, BSF grant 2008477, and Sloan and Okawa fellowships.}}

\maketitle

\begin{abstract}
Consider an optimization problem with $n$ binary variables and $d+1$ linear objective functions.  Each valid solution $x \in \{0,1\}^n$ gives rise to an objective vector in $\R^{d+1}$, and one often wants to enumerate the Pareto optima among them.  In the worst case there may be exponentially many Pareto optima; however, it was recently shown that in (a generalization of) the smoothed analysis framework, the expected number is polynomial in~$n$.  Unfortunately, the bound obtained had a rather bad dependence on $d$; roughly $n^{d^d}$.  In this paper we show a significantly improved bound of~$n^{2d}$.  

Our proof is based on analyzing two algorithms.  The first algorithm, on input a Pareto optimal $x$, outputs a ``testimony'' containing clues about $x$'s objective vector, $x$'s coordinates, and the region of space $B$ in which $x$'s objective vector lies.  The second algorithm can be regarded as a {\em speculative} execution of the first --- it can uniquely reconstruct $x$ from the testimony's clues and just \emph{some} of the probability space's outcomes.  The remainder of the
probability space's outcomes are just enough to bound the probability that $x$'s objective vector falls into the region~$B$.
\end{abstract}

\setcounter{page}{0} \thispagestyle{empty}
\newpage

\section{Introduction}

We study the expected number of Pareto optimal solutions in multiobjective binary optimization problems within the framework of \emph{smoothed analysis}.

\subsection{Multiobjective optimization and Pareto optima}

In a typical decision-making problem there are multiple criteria used in judging the quality of a solution.  For example, in choosing a driving route between two points one might want to minimize distance, tolls, number of turns, and expected traffic; in choosing a vacation hotel one might want to minimize price and distance to the beach, while maximizing quality.  In such cases there is rarely a single solution which is best on all criteria simultaneously.  The most popular way to handle the tradeoff is to determine the set of all \emph{Pareto optimal} solutions, meaning those solutions which are not dominated in all measures of quality by some other solution. This idea, originating in microeconomics, has been very extensively studied in computer science, especially in operations research~\cite{Ehr05}, algorithmic theory~\cite{PY02}, artificial intelligence~\cite{Deb01}, and database theory (under the name ``skyline queries'')~\cite{BKS01}.

Even if one is not interested in Pareto optima per se, many algorithms and heuristics for solving optimization problems enumerate Pareto optimal solutions as an intermediary step. For example, the Nemhauser--Ullmann algorithm~\cite{NU69} for exactly solving the Knapsack problem works by iteratively computing the Pareto optimal $\langle \text{value},\text{weight} \rangle$ pairs achievable by the first $i$ items, for $i = 1 \dots n$. Beier and V\"ocking~\cite{BV04b} observed that this algorithm runs in time $O(nB)$, where $B$ is an upper bound on the number of Pareto optima in each stage. As we describe below, this allowed them to give the first polynomial-time algorithm for an NP-hard optimization problem in the smoothed analysis framework, by performing smoothed analysis on the number of Pareto optimal solutions.

Unfortunately, even in the simplest case multiobjective optimization --- two linear objective functions --- the number of Pareto optimal solutions may be exponentially large in the number of decision variables.  There have been two main approaches to dealing with this exponential complexity.  The first, popularized by Papadimitriou and Yannakakis~\cite{PY02}, involves computing ``$\eps$-approximate Pareto sets''.  In many cases, polynomial-size $\eps$-approximate Pareto sets can be computed efficiently; see the thesis of Diakonikolas~\cite{Dia10} for references. The second approach, pioneered by Beier and V\"ocking~\cite{BV04b}, involves studying multiobjective optimization in the smoothed analysis framework.

\subsection{Smoothed analysis for discrete optimization}
Smoothed analysis was introduced in an influential work of Spielman and Teng~\cite{ST04}, as a framework intermediate between worst-case and average-case analysis.  Here the idea is to think of real numbers in the input as being defined based on imprecise measurements; specifically, they are first fixed adversarially in $[-1,1]$, say, and then subjected to Gaussian noise with some small standard deviation~$\sigma$.  In this framework, Spielman and Teng showed that a certain version of the simplex algorithm for linear programming runs in $\poly(n, 1/\sigma)$ expected time.

A notable work of Beier and V\"ocking~\cite{BV04b} from 2003 showed that the NP-hard $0/1$-Knapsack problem can be solved in polynomial time in the smoothed analysis framework.  (Previously, there had been a long line of work on average-case analysis of $0$/$1$-Knapsack: see, e.g.,~\cite{DF89,GM84b,Lue98}.)  Furthermore, they showed this holds even in a much more general model of smoothed analysis.  In one version of their model, each item's profit $P_i$ and weight $W_i$ is an independent random variable with values in $[0,1]$; the only restriction is that the probability density function (pdf) of each $P_i$ and $W_i$ is upper-bounded by the parameter~$\pdfbd$.  We call this model ``$\pdfbd$-semirandom''.  It is easy to see that as $\pdfbd$ is increased, the framework goes from (a very general version of) average-case analysis to worst-case analysis.  For example, given a small number $\sigma$, if we take $\pdfbd = 1/\sigma$ then the profits $P_i$ could be of the form $p_i + U_i$, where $p_i \in [\sigma, 1-\sigma]$ is an adversarially chosen number and $U_i$ is uniformly random on $[-\sigma, \sigma]$. (The original case of Gaussian noise does not quite fit in this framework, but is easily handled with a small amount of additional work.)

\subsection{Previous work}
Beier and V\"ocking showed that in this $\pdfbd$-semirandom model, the expected number of Pareto optimal knapsacks is $O(\pdfbd n^4)$; from this they immediately deduced that the Nemhauser--Ullmann algorithm runs in expected $O(\pdfbd n^5)$ time.  In fact, Beier and V\"ocking showed that the same is true even if the weights are \emph{adversarially} specified, and only the profits are chosen randomly (independently, from $\pdfbd$-bounded distributions).  In this case of adversarially weights, they also showed an $\Omega(n^2)$ lower bound for the expected number of Pareto optima, even for uniformly distributed profits (i.e., $\pdfbd = 1$).

In his thesis, Beier~\cite{Bei04} extended this analysis to general $2$-objective binary optimization problems.  Specifically, he showed that given an arbitrary set of ``solutions'' $\allitems \subseteq \{0,1\}^n$ and arbitrary $2$nd objective values $\objvec{x}{2}$ for each $x \in \allitems$, if the $1$st objective is linear and $\pdfbd$-semirandom, then the expected number of Pareto optima is still $O(\pdfbd n^4)$.  Later work of Beier, R\"oglin, and V\"ocking~\cite{BRV07} improved this bound to $O(\pdfbd n^2)$ (which is tight for constant $\pdfbd$) and also extended to the case of integer-valued decision variables.

These works only handled the case of $2$ objectives.  Recently, R\"oglin and Teng~\cite{RT09} extended the analysis to the case in which there are $d+1$ objective functions, $d$ of which are linear and $\pdfbd$-semirandom, and one of which is completely arbitrary.  Their bound on the expected number of Pareto optima is polynomial in $n$ and $\pdfbd$ for constant~$d$, and they were also able to polynomially bound higher moments. Unfortunately, their result is probably of theoretical interest only, as the dependence on $d$ is rather bad.  E.g., for $d = 3$ their upper bound on the expected number of Pareto optima is roughly $n^{97}$ assuming $n \geq 2^{453787938}$ (and is much worse than $n^{97}$ for smaller $n$).  In general their bound is roughly $(\sqrt{\pdfbd} \cdot n)^{f(d)}$ for $f(d) = 2^{d-1} (d+1)!$, once $n \geq \exp(\exp(d^2 \log d))$.  R\"oglin and Teng concluded their work by asking whether the exponent $f(d)$ on $n$ could be reduced from $d^{\Theta(d)}$ to $\poly(d)$; this was later recognized as an important open problem~\cite{Ten01}. Here, we resolve this question.

Very closely related to the research we have just described is a sequence of works~\cite{BV06,ANRV07,RV07,RT09}, starting with Beier and V\"ocking and culminating with R\"oglin and Teng, showing that binary optimization problems are solvable in expected polynomial time in the smoothed analysis framework if and only if they are solvable in randomized pseudopolynomial time in the worst case.

\subsection{Our contribution}

In this work we give an affirmative answer to the open problem of R\"oglin and Teng, reducing their bound from roughly $n^{2^{d-1} (d+1)!}$ to $n^{2d}$.  Thus the exponent on $n$ can in fact be made \emph{linear} in $d$. In particular, we prove that the expected number of Pareto optimal \itms in the model described above is at most $2 \cdot (4\pdfbd d)^{d(d+1)/2} \cdot n^{2d}$. It is interesting to compare our result with what is known about Pareto optima when $2^n$ points are chosen independently and uniformly in $[-1,1]^{d+1}$. In this scenario, old results~\cite{BKST78,Dev80,Buc89} show that the expected number of Pareto optima is $\Theta(n)^d$ for each constant~$d$.  Our bound is within a square of this quantity, despite the significant dependencies in the model.  We also note that this square is necessary at least for $d = 1$, due to the $\Omega(n^2)$ lower bound of Beier and V\"ocking~\cite{BV04b}.

Usually, in smoothed analysis we are interested in demonstrating that a certain algorithm runs quickly or that a certain approximation algorithm returns a near-optimal solution. In such cases, one often defines an event -- some property of the data that ensures an algorithm runs quickly or an approximation algorithm works well. This is true in the context of previous literature on the number of Pareto optimal \itms as well --- indeed, the works of \cite{BV06,RT09} are based on notions of \emph{winner gap} and \emph{loser gap} which can be interpreted as a discrete analogue to condition number. 

Our approach turns this around: We give a deterministic algorithm, which on input a Pareto optimal $x$, runs on the data and produces an event -- in the form of a ``testimony'' containing clues about $x$'s objective vector, $x$'s coordinates, and the region of space $B$ in which $x$'s objective vector lies. Our family of events is rather complicated, but is defined implicitly based on a simple algorithm. 

We then give a second algorithm which can be regarded as a {\em speculative} execution of the first --- it can uniquely reconstruct $x$ from the testimony's clues and just \emph{some} of the probability space's outcomes.  The remainder of the
probability space's outcomes are just enough to bound the probability that $x$'s objective vector falls into the region~$B$. So we are able to bound the probability that any particular "testimony" is output by the first algorithm, and consequently we are able to give an upper bound on the expected number of Pareto optimal \itms.

\section{Our result and approach}

In this section we will describe the problem formally, state our Main Theorem, and then briefly describe our approach.  The remainder of the paper is devoted to the proof of the Main Theorem.

\subsection{Problem definitions}

Our setting captures the broad class of multiobjective binary optimization problems in which all (but one) of the objective functions are linear. We fix once and for all an arbitrary set $\allitems \subseteq \zo^n$ of \emph{\itms{}}.  These might encode knapsacks, the sets of edges forming a spanning tree in a given graph, or even the sets of edges forming a Hamiltonian cycle.

\paragraph{Matrix notation.} We think of \itms in $\allitems \subseteq \zo^n$ as column vectors.  For a matrix (or vector) $A$, we will write $A^i$ for the $i$'th row of $A$ and write $A_j$ for the $j$'th column of $A$; thus $A^i_j$ is the $(i,j)$ entry of $A$.  For $i < k$ we will also write $A^{i..k}$ for the submatrix of $A$ consisting of rows $i$ through $k$.  Given matrices $A$ and $B$ of the same size we write $A \circ B$ for their Hadamard product, i.e., their entry-wise product.  Thus $(A \circ B)^i_j = A^i_j B^i_j$.

\paragraph{\Weights and \objs.}   Associated to each \itm $x \in \allitems$ are $d+1$ \emph{\objs{}}; we encode them with a column vector $\objvec{x}{} \in \R^{d+1}$. The first $d$ \objs are assumed to be linear and are chosen in a ``$\pdfbd$-semirandom'' fashion.  More specifically, there is a $d \times n$ matrix $\bWts$ of random variables called \emph{\weights{}}. (We typically write random variables in boldface.)  We assume that each entry of $\bWts$ is an independent, continuous random variable with support on $[-1,1]$ and pdf bounded by the parameter~$\pdfbd$.  The first $d$ \objs of \itm  $x$ are defined by the equation $\objvec{x}{1..d} = \bWts x$.  (Recall that $x \in \zo^n$ is thought of as a column vector.)
The $(d+1)$'th \objs of the \itms are neither linear nor random.  We assume merely that they are fixed, distinct real numbers, chosen in advance of $\bWts$.  (Indeed, their magnitudes are not important for us, only their relative ordering.) We will significantly abuse notation by writing $\bWts^{d+1} x$ in place of $\objvec{x}{d+1}$. In this way, $\objvec{x}{i} = \bWts^i x$ holds for each $i \in [d+1]$.

\paragraph{Pareto optima.}  Without loss of generality, we think of higher \objs as preferable.  Accordingly, given (column) vectors $p, q \in \R^{d+1}$ we say that $p$ \emph{dominates}~$q$ if $p \geq q$.  Here ``$\geq$'' is to be interpreted entry-wise when applied to vectors; i.e., $p$ dominates~$q$ if $p^i \geq q^i$ for all $i \in [d+1]$.  We will also sometimes use the notion of \emph{$t$-domination} for $t < d+1$; we say that $p$ $t$-dominates $q$ if $p^{1..t} \geq q^{1..t}$.  Given a set of points $\calP \subset \R^{d+1}$ we say that $p \in \calP$ is \emph{Pareto optimal (within $\calP$)} if $p$ is not dominated by any other point $q \in \calP$; i.e., for all $q \in \calP \setminus \{p\}$, there exists $i \in [d+1]$ with $p^i > q^i$.  Of course, we will be interested in applying this concept to the \objs of the \itms in $\allitems$.  Given $\bWts$, we consider $\calP = \{\objvec{z}{} : z \in \allitems\} \subset [-n,n]^d \times \R$.  We then say that the \emph{\itm{}} $x \in \allitems$ is Pareto optimal if $\objvec{x}{}$ is Pareto optimal within $\calP$.  Finally, given $\bWts$, we define $\PO \subseteq \allitems$ to be the set of all Pareto optimal \itms.

\subsection{Our result} \label{sec:maintheorem}
We can now state our Main Theorem:

\begin{named}{Main Theorem} \qquad $\displaystyle \Ex_{\bWts}\Bigl[\bigl|\PO\bigr|\Bigr] \leq 2 \cdot (4\pdfbd d)^{d(d+1)/2} \cdot n^{2d}.$
\end{named}

\subsection{Our approach} \label{sec:boxdef}
To prove the Main Theorem we use a probabilistic argument which has a rather unusual form.  Unfortunately, it is also fairly intricate. In this section we will try to convey some of the ideas of the argument while hiding a number of complicating details.

Our proof can be seen as a $d$-dimensional generalization of the Beier--R\"oglin--V\"ocking $O(\pdfbd n^2)$ upper bound for the $d = 1$ case (which we will later sketch).  The reader is advised to keep the cases $d = 1,2$ in mind for visualization purposes.  Recall that the \itms $x \in \allitems$ have $d$ semirandom linear \objs but their $(d+1)$'th \objs are fixed in advance arbitrarily.  Once the \weights $\bWts$ are drawn and the \objs $\objvec{x}{1..d} \in [-n,n]^d$ thus determined, one can think of identifying the Pareto optima among $\allitems$ via a ``sweep'' along the $(d+1)$'th dimension.  This means proceeding through the \itms $x \in \allitems$ in decreasing order of $\objvec{x}{d+1}$ and considering the ``point'' $\objvec{x}{1..d} = \bWts x \in [-n,n]^d$; the set of points which are not $d$-dominated by any previously seen point correspond exactly to the set of Pareto optimal \itms.

\paragraph{\Boxes and density.} An oversimplification of our proof is to think of it as showing that the ``probability density'' of Pareto optimal points in $[-n,n]^d$ is not too high; roughly $O(n^d)$.  In aid of making this formal, we fix once and for all a real number $\eps > 0$ which should be thought of as extremely small, $\eps \ll 1/(\pdfbd d 2^{2n})$.  Additionally, we assume that $1/\eps$ is an integer.  We then introduce the following definition:
\begin{definition} For a point $b \in (\eps \Z)^d$, we define the \emph{\tbox{d} based at point~$b$} to be $b + [0,\eps)^d$.  Note that the set of all \tboxes{d} partitions $[-n,n)^d$ and indeed all of $\R^d$.  More generally, for $t \in [d]$ and $b \in (\eps \Z)^j$, we define the \emph{\tbox{t} based at point~$b$} to be $(b + [0,\eps)^d) \times \R^{d-t}$.  The set of all \tboxes{t} also partitions $\R^d$.
\end{definition}
Since $\eps$ is so small, the probability that there will be two different points $\bWts x$ and $\bWts x'$ in the same \tbox{d} is negligible.  Thus if $B$ denotes an arbitrary \tbox{d}, we can upper-bound the number of Pareto optima by $(2n/\eps)^d$ times the probability that there is a Pareto optimum $x \in \allitems$ with $\objvec{x}{1..d}$ in $B$.  Our goal is to bound this probability by roughly $O(n^d) \eps^d$.

In particular, we must make sure to keep the probability roughly comparable to $\eps^d$.  A crucial aspect of our proof is that we can bound $\Pr[\bWts x \in B]$ by $(\pdfbd \eps)^d$ for any $x \neq \vec{0}$ \emph{while only using a small part of the probability space}.  Specifically, suppose we select $j \in [n]$ such that $x^j \neq 0$, and then imagine drawing all entries of $\bWts$ except for the $j$'th column $\bWts_j$.  Then the final position of the point $\bWts x$ is of the form $(p^1 + \bWts^1_j, \dots, p^n + \bWts^n_j)$, where the $p^i$'s are constants.  This point will lie in the \boxx $B$ only if each \weight $\bWts^i_j$ falls into a certain fixed interval of width $\eps$.  As the random variables $\bWts^i_j$ are independent and have pdf's bounded by $\pdfbd$, the probability that all $\bWts^i_j$'s fall into the required intervals is at most $(\pdfbd \eps)^d$.  Note that this argument works for any possible outcome of the $d(n-1)$ \weights outside of $\bWts_j$.

\paragraph{Uniqueness.} Unfortunately we cannot simply take this observation and union-bound over all potential Pareto optima $x$, since this would lose a factor of $|\allitems|$.  We would be in much better shape if, after all \weights except for $\bWts_j$ were drawn, there were very few \itms $x$ --- or even just a unique \itm $x$ --- for which the event
\[
\T_x = \text{``$x$ is Pareto optimal with $\bWts x \in B$''}
\]
had a chance of occurring.  Here by ``have a chance of occurring'', we mean $\Pr_{\bWts_j}[\T_x] > 0$.   In the simplest case of $d = 1$, Beier, R\"oglin, and V\"ocking~\cite{BRV07} essentially show that essentially holds if one adds some extra conditions to the event $\T_x$.  We now sketch a reinterpretation of their argument.

\paragraph{The Beier--R\"oglin--V\"ocking argument.} Note that since $d = 1$ for this sketch, the \weights matrix $\bWts$ is just a random (row) vector. For each $j \in [n]$ and \tbox{1} (interval) $B$, let us define the significantly more complicated event
\[
\T_{x, j, B} = \text{``$x^j = 1$, $\bWts x \in B$, $x$ is Pareto optimal, and the `next' Pareto optimum $y$ has $y^j = 0$''}.
\]
Here `next' refers to the ``sweep along the $2$nd coordinate''; i.e., $y$ is the \itm $z$ with maximal $\objvec{z}{2}$ among $\{z \in \allitems : \bWts z > \bWts x\}$.  The Beier--R\"oglin--V\"ocking argument takes a union bound over all $j \in [n]$ in addition to over all $B$.  The key to their argument is the following ``uniqueness'' claim:  for any draw of the \weights other than $\bWts_j$, there is a \emph{unique} $x \in \allitems$ for which the event $\T_{x,j,B}$ has a chance of occurring.  Given this claim, the proof is almost complete.  For that unique $x$ the event $\T_{x,j,B}$ still has at most a $\pdfbd \eps$ chance of occurring, since $x^j$ must be $1$ and the \weight $\bWts_j$ is still independent and undrawn.  Union-bounding over all $j$ and $B$, one concludes that the expected value of 
\[
\#\{\text{Pareto optimal $x$} : \text{the `next' Pareto optimum $y$ has $y^j \neq 1 = x^j$ for some $j$}\}
\]
is at most $n \cdot (2n/\eps) \cdot (\pdfbd \eps) = 2\pdfbd n^2$.  This \emph{almost} counts the total number of Pareto optima. Certainly for each Pareto optimum $x$, there is \emph{some} coordinate $j$ such that the `next' Pareto optimum $y$ has $y^j \neq a = x^j$; it's just that this bit $a$ might be $0$ rather than $1$.  The Beier--R\"oglin--V\"ocking is concluded (essentially) by union-bounding over $a \in \{0,1\}$ as well.  (It may seem crucial that $x^j$ was $1$ and not $0$ when we observed that $\Pr_{\bWts_j}[\bWts x \in B] \leq \pdfbd \eps$.  This difficulty is overcome with an additional trick, changing the condition $\bWts x \in B$ in $\T_{x,j,B}$ to the condition $\bWts x - \bWts_j \overline{a}$ in $\T_{x,j,a,B}$.)

\paragraph{The R\"oglin--Teng argument.}  How can we generalize this argument to $d$ dimensions?  R\"oglin and Teng roughly take the following approach.  First, they generalize the above argument to show that for $d = 1$, the expected $c$'th power of the number of Pareto optima is $(\pdfbd n^2)^{c(1+o(1))}$.  This gives them a concentration result for the number of Pareto optima.  They then proceed by induction on the dimension $d$.  In reducing from dimension $d$ to $d-1$ there are two difficulties.  First, instead of having a unique $x$ to worry about as in the Beier--R\"oglin--V\"ocking,  they need to worry about all \itms in a $(d-1)$-dimensional Pareto set.  One expects this not to be too large, by induction; however, their argument needs a high-probability result.  Hence they need to inductively bound higher powers of the number of Pareto optima.  This induction leads to the rather bad dependence on $d$.  A second difficulty they face comes from their use of conditioning to separate the $d$'th dimension from the first $d-1$; this introduces dependencies that they must work to control.

\paragraph{Our argument.} We define a family of events $\T_{x,\J,\A,\BoxL}$.  These events are again of the form ``$x$ falls into a \boxx related to $\BoxL$ and certain other lower-dimensional conditions happen''.  We need to define these other conditions in an extremely careful way so that the following holds:
\begin{center}
\fbox{\parbox{5.0in}{
  \emph{Based on $\J$, there is a way to partition the draw of $\bWts$ into two parts called $\bWJbar$ and $\bWJ$.  Given the outcome of $\bWJbar$, there is a {\bf unique} $x \in \allitems$ for which $\T_{x,\J,\A,\BoxL}$ can occur.  Furthermore, the randomness remaining in $\bWJ$ is such that the probability of $\T_{x,\J,\A,\BoxL}$ can {\bf still be bounded} by an appropriately small quantity.}}}
\end{center}

We manage to identify the necessary conditions; however they are complicated enough that they cannot be described with just a sentence of text.  Instead, we come to the first unusual aspect of our argument; the extra conditions are of the form ``\emph{a certain deterministic algorithm $\Trans$, when run with input $x$ and $\bWts$, produces a certain output \transcript{}}''.  At this point the reader might think that such conditions have no chance to satisfy the boxed property above: in particular, since $\Trans$ depends on $\bWts$, how can knowing its output still leave the $\bWJ$ part of the probability space free?  We overcome this problem with a second unusual idea.  We introduce \emph{another} deterministic algorithm called $\Recon$, which takes as input the output $\Trans(x, \bWts)$, along with the outcome of $\bWJbar$.  We show that using just this information, $\Recon$ can recover the input $x$, \emph{assuming that it is Pareto optimal}.  This ability to reverse-engineer $x$ gives us the needed ``uniqueness'' property, and the fact that $\Recon$ does not need to know $\bWJ$ -- and yet this amount of remaining randomness is still enough to bound the probability that $x$ falls into certain \boxes.

\section{Outline of the proof} \label{sec:outline}

At this point we move from intuition to precise details.  In this section we give the overall structure of our proof of the Main Theorem.  By the end of this outline we will have reduced it to a number of lemmas, which are then proven in the appendices of the paper.

\subsection{\Transcripts}
The first key ingredient in our proof is a deterministic map we call $\Trans$, which takes as input \anitem $x \in \allitems$ and a fixed matrix of \weights $\Wts$, and outputs a ``\transcript'' $(\J, \A, \BoxL)$:
\[
\Trans \co (x,\Wts) \mapsto (\J,\A,\BoxL).
\]
(The map $\Trans$ also depends on the fixed quantities $n$, $\eps$, $\allitems$, and the $(d+1)$'th \objs $\objvec{z}{d+1}$.) We will actually only care about the behavior of $\Trans(x,\Wts)$ when the \weights $\Wts$ make $x$ into a Pareto optimum, but it is clearest to define the mapping for every pair of $x$ and $\Wts$.

Regarding the \transcript itself, roughly speaking $\J$ is a list of $d$ coordinates in $[n]$, $\A$ is a ``diagonalization matrix'' consisting of $d$ bits per coordinate in $\J$, and $\BoxL$ is a list of \tboxes{t}, one for each $t \in [d]$.  \emph{Very} roughly speaking, the meaning of $\Trans(x,\Wts) = (\J,\A,\BoxL)$ is that the bits $\{x^j : j \in \J\}$ agree with certain bits in $\A$ and that $\Wts x$ falls into the \boxes in $\BoxL$ --- or rather, that a slight translation of $\Wts x$ based on $\A$ falls into these \boxes.  Precise details are given in Section~\ref{sec:trans}, but they are not important for understanding the outline of the proof.  On first reading, one should think of the number of possible \transcripts as something roughly like $n^{2d} / \eps^{d(d+1)/2}$.

\subsection{The \ok event}
We will also need to define a simple event based on the random draw of $\bWts$ which we call $\OK$. In studying Pareto optima we prefer not to distinguish between domination and ``strict'' domination.  Luckily we don't have to: since the entries of $\bWts$ are continuous random variables, the probability that $\bWts^i x = \bWts^i y$ for any $i \in [d]$ and distinct $x, y \in \allitems$ is~$0$.  Our event $\OK$, which we now formally define, slightly generalizes this:
\begin{definition}
$\OK = \OK(\bWts)$ is defined to be the event that $|\bWts^i x - \bWts^i y| > \eps$ for all $i \in [d]$ and distinct $x,y \in \allitems$.
\end{definition}
We require the following simple lemma:
\begin{named}{\ok Lemma}  $\Pr[\neg \OK] \leq \pdfbd d 2^{2n+1} \eps$.
\end{named}
\begin{proof}  For each fixed $i \in [d]$ and distinct $x, y \in \zo^n$, we show that $\Pr[|\bWts^i x - \bWts^i y| \leq \eps] \leq 2 \pdfbd \eps$; the result then follows by a union bound. Since $x$ and $y$ are distinct we may select $j \in [n]$ such that $x^j - y^j = 1$, after possibly exchanging $x$ and $y$.  Now imagine that the \weights $\{ \bWts^i_k : k \neq j\}$ are drawn first; then the event $|\bWts^i x - \bWts^i y| \leq \eps$ becomes of the form $|c + \bWts^i_j| \leq \eps$ for some constant~$c$.  By independence, the random variable $\bWts^i_j$ still has pdf bounded by $\pdfbd$, so this event has probability at most $\pdfbd \cdot 2\eps$, as desired.
\end{proof}

\subsection{Proof of the Main Theorem}
We are now able to outline the proof of the Main Theorem.

\begin{equation} \label{eqn:ok-split-up}
\Ex_{\bWts}\Bigl[\bigl|\PO\bigr|\Bigr] = \Ex_{\bWts}\Bigl[\bigl|\PO\bigr| \cdot \bone{\OK}\Bigr] + \Ex_{\bWts}\Bigl[\bigl|\PO\bigr| \cdot \bone{\neg \OK}\Bigr].
\end{equation}
Regarding the second term in~\eqref{eqn:ok-split-up}, naively we have
\begin{equation} \label{eqn:not-ok-term}
\Ex_{\bWts}\Bigl[\bigl|\PO\bigr| \cdot \bone{\neg \OK}\Bigr] \leq \Ex_{\bWts}\Bigl[2^n \cdot \bone{\neg \OK}\Bigr] = 2^n \Pr[\neg \OK]\leq \pdfbd d 2^{3n+1} \eps,
\end{equation}
using the \ok Lemma.  As for the first (main) term in~\eqref{eqn:ok-split-up}, we break it up according to the possible \transcripts:
\begin{equation}  \label{eqn:trans-break-up}
\Ex_{\bWts}\Bigl[\bigl|\PO\bigr| \cdot \bone{\OK}\Bigr] = \sum_{(\J, \A, \BoxL)} \Ex_{\bWts}\left[\littlesum_{x \in \allitems} \bone{x \in \PO} \cdot \bone{\Trans(x,\bWts) = (\J,\A,\BoxL)} \cdot \bone{\OK}\right].
\end{equation}

For a given draw of \weights $\bWts$, it is possible to show that \emph{if} the event $\OK$ occurs, then the different $x \in \PO$ generate unique \transcripts $(\J, \A, \BoxL)$.  (This follows from the \Transcript--Determines--PO Lemma in Section~\ref{sec:trans}.) In other words, for a fixed \transcript $(\J,\A,\BoxL)$, after $\bWts$ is drawn there can be at most one $x \in \allitems$ for which the event
\[
\bigl(x \in \PO\bigr) \wedge \bigl(\Trans(x,\bWts) = (\J,\A,\BoxL)\bigr) \wedge \OK
\]
occurred.  This shows that~\eqref{eqn:trans-break-up} is at most the number of possible \transcripts.  Unfortunately, that is not a helpful bound because the number of possible \transcripts includes the huge factor $(1/\eps)^{d(d+1)/2}$.

We now come to the key idea in the proof.  For each fixed \transcript $(\J,\A,\BoxL)$, we split up the draw of $\bWts$ into two stages in a way that depends on~$J$.  In the first stage, ``most'' of the $dn$ entries of $\bWts$ are drawn; we denote these entries by $\bWJbar$ for reasons to be explained later.  In the second stage, the remaining ``few'' entries of $\bWts$ are drawn (independently, of course); we denote this second set of entries by $\bWJ$.  On first reading, one should think of ``few'' as meaning $d(d+1)/2$.  Now the key idea is that the uniqueness property described above holds \emph{even after just drawing $\bWJbar$}:
\begin{named}{Uniqueness Lemma}
Fix a \transcript $(\J,\A,\BoxL)$ and fix the outcome of $\bWJbar$.  Then there exists at most one $x \in \allitems$ such that the event
\[
\bigl(x \in \PO\bigr) \wedge \bigl(\Trans(x,\bWts) = (\J,\A,\BoxL)\bigr) \wedge \OK
\]
can occur.  Here the event's randomness is just the draw of $\bWJ$.
\end{named}
Based on this idea, we write~\eqref{eqn:trans-break-up} as
\[
\sum_{(\J, \A, \BoxL)} \Ex_{\bWJbar}\left[\littlesum_{x \in \allitems}\ \  \Pr_{\bWJ}\left[\bigl(x \in \PO\bigr) \wedge \bigl(\Trans(x,\bWts) = (\J,\A,\BoxL)\bigr) \wedge \OK\right]\right].
\]
The Uniqueness Lemma says that for each choice of $(\J, \A, \BoxL)$ and $\bWJbar$, at most one of the probabilities in the above expression can be nonzero.  Hence we may upper-bound~\eqref{eqn:trans-break-up} by
\begin{equation} \label{eqn:used-uniqueness}
\sum_{(\J, \A, \BoxL)} \Ex_{\bWJbar}\left[\max_{x \in \allitems}\ \  \Pr_{\bWJ}\left[\bigl(x \in \PO\bigr) \wedge \bigl(\Trans(x,\bWts) = (\J,\A,\BoxL)\bigr) \wedge \OK\right]\right].
\end{equation}

We now complete the proof by showing that there is enough randomness left in $\bWJ$ so that for any $x \in \allitems$, even the probability of the event $\Trans(x,\bWts) = (\J, \A, \BoxL)$ is small.  We bound this probability in terms of a parameter called $\dim(\BoxL)$, which we define in Section~\ref{sec:trans}  For now, it suffices to know that $\dim(\BoxL)$ is an integer between $0$ and $d(d+1)/2$; on first reading, one should think of it as simply always being $d(d+1)/2$.
\begin{named}{Boundedness Lemma} For every fixed $(\J,\A,\BoxL)$, outcome of $\bWJbar$, and $x \in \allitems$, it holds that
\[
\Pr_{\bWJ}[\Trans(x,\bWts) = (\J,\A,\BoxL)] \leq \pdfbd^{\dim(\BoxL)} \eps^{\dim(\BoxL)}.
\]
\end{named}
\noindent Using this in~\eqref{eqn:used-uniqueness} we upper-bound~\eqref{eqn:trans-break-up} by
\begin{equation} \label{eqn:used-boundedness}
\sum_{(\J, \A, \BoxL)} \pdfbd^{\dim(\BoxL)} \eps^{\dim(\BoxL)}.
\end{equation}

As mentioned, on first reading one should think of $\dim(\BoxL)$ as always being $d(d+1)/2$ and one should think of the number of possible \transcripts as being roughly $n^{2d}/\eps^{d(d+1)/2}$.  Thus~\eqref{eqn:used-boundedness} is roughly $\pdfbd^{d(d+1)/2} \cdot n^{2d}$, comparable to the quantity in the Main Theorem.  We will eventually do a more precise but straightforward estimation to bound~\eqref{eqn:used-boundedness} (and hence~\eqref{eqn:trans-break-up}):
\begin{named}{Counting Lemma}  For a fixed $n$ and $\eps$,
\[
\sum_{\substack{\text{\emph{possible \transcripts{}}} \\ (\J, \A, \BoxL)}} \pdfbd^{\dim(\BoxL)} \eps^{\dim(\BoxL)} \leq
2 \cdot (4d\pdfbd)^{d(d+1)/2} \cdot n^{2d}.
\]
\end{named}
Substituting this bound on~\eqref{eqn:trans-break-up}, as well as the bound~\eqref{eqn:not-ok-term}, into~\eqref{eqn:ok-split-up} yields
\[
\Ex_{\bWts}\Bigl[\bigl|\PO\bigr|\Bigr]  \leq 2 \cdot (4d\pdfbd)^{d(d+1)/2} \cdot n^{2d} + \pdfbd d 2^{3n+1} \eps.
\]
Since we can make $\eps$ arbitrarily small, the proof of the Main Theorem is complete.

\section{\Transcripts} \label{sec:trans}

In this section we describe the $\Trans$ algorithm, which assumes $n$, $\eps$, $\allitems$, and the $(d+1)$'th \objs $\objvec{z}{d+1}$ are fixed. The input to $\Trans$ is \anitem $x \in \allitems$ and a $d \times n$ matrix of \weights $\Wts$.  The output is a ``\transcript'', which is a triple $(\J, \A, \BoxL)$.

\subsection{Components of a \transcript}
We now describe the components of a \transcript.

\paragraph{Index vector.} We call the first component, $\J$, the ``index vector''.  This is defined to be a length-$d$ row vector from $([n] \cup \{\bot\})^d$ in which all non-$\bot$ indices are distinct.  On first reading, one should ignore the possibility of $\bot$'s and simply think of an index vector $\J$ as an ordered list of $d$ distinct indices from $[n]$.

\paragraph{Diagonalization matrix.} We call the second component, $\A$, a ``diagonalization matrix''.  $\A$ is $n \times d$ matrix with entries from $\{0, 1, \bot\}$.  Most entries in $\A$ will be $0$; indeed, the row $\A^j$ will be nonzero only if $j$ is one of the indices in $\J$.  Before describing $\A$ completely formally, let us describe the ``typical'' case when $\J$ contains no $\bot$'s, and thus just consists of distinct indices from $[n]$.  In this case, $\A$ will also contain no $\bot$'s.  To make the picture even clearer, let us imagine that $\J$ is simply $(1, 2, \dots, d)$.  Thus $\A$ will only be nonzero in its first $d$ rows, so let us write $\A' = \A^{1..d}$. In this case, if $x \in \allitems$ is the input to $\Trans$, then $\A'$ will be of the following form:
\[
\left[ \begin{array}{cccccc}
            \overline{x^{1}} & \ast & \ast & \ast & \cdots & \ast \\
            x^{2} & \overline{x^{2}} & \ast & \ast & \cdots & \ast \\
            x^{3} & x^{3} & \overline{x^{3}} & \ast & \cdots & \ast \\
            x^{4} & x^{4} & x^{4} & \overline{x^{4}} & \cdots & \ast \\
             \vdots & \vdots & \vdots & \vdots & \ddots & \vdots \\
            x^{d} & x^{d} & x^{d} & x^{d} & \cdots & \overline{x^{d}}
       \end{array} \right].
\]
Here each $x^j$ is of course in $\zo$, we write $\overline{x^j}$ for $1 - x^j$, and $\ast$ denotes that the entry may be either $0$ or~$1$.  We say that \emph{$\A$ diagonalizes $x$ on $\J = (1, 2, \dots, d)$}.  We now give the formal definition which includes the possibility of $\J$ containing $\bot$'s.
\begin{definition}
  Given an index vector $\J$ and \anitem $x \in \zo^n$, we say that the matrix $\A \in \{0,1,\bot\}^{n \times d}$ is a \emph{diagonalization matrix}, and in particular that it \emph{diagonalizes $x$ on $\J$}, if the following conditions hold:  If $j \in [n]$ does not appear in $\J$, then row $\A^j$ is all zeros.  Otherwise, if $j = \J_u \in [n]$ for some $u \in [d]$:
  \begin{equation} \label{eqn:diag-conds}
    A^j_t = \bot \text{ if and only } \J_t = \bot, \qquad
    A^j_u = \overline{x^j}, \qquad
    A^j_t = x^j \text{ for all } t < u \text{ with $\J_t \neq \bot$.}
  \end{equation}
\end{definition}

\paragraph{\Boxx list.} The last component of a \transcript, $\BoxL$, is a list $\BoxL = (\Bx_1, \dots, \Bx_d)$.  For $t \in [d]$ we have that $\Bx_t = \bot$ if $\J_t = \bot$, and otherwise $\Bx_t$ is a \tbox{t}, as defined in Section~\ref{sec:boxdef}.  We define the \emph{\dimension} of the \boxx list $\BoxL$ to be $\sum \{t \in [d] : \Bx_t \neq \bot\}$.  On first reading, one should ignore the possibility of $\Bx_t = \bot$, in which case $\dim(\BoxL)$ is always $1 + 2 + \dots + d = d(d+1)/2$.

\paragraph{\Masking matrix.} Having defined the components $(\J,\A,\BoxL)$ of a \transcript, we now explain one more piece of notation; that of a \masking matrix.  Given an index vector $\J$, we define the associated \masking matrix $\free{\J} \in \{0,1\}^{d \times n}$ as follows:
\[
\free{\J}^i_j = \begin{cases}
                    1 &  \text{if $j = \J_t \in [n]$ for some $t \in [d]$ and $i \leq t$,}\\
                    0 & \text{otherwise.}
                \end{cases}
\]
For illustration, if $\J = (1, 2, \dots, d)$, then $\free{\J}$ is the mostly-zeros $d \times n$ matrix whose left-most $d \times d$ submatrix is
\[
\left[ \begin{array}{cccccc}
            1 & 1 & 1 & 1 & \dots & 1 \\
            0 & 1 & 1 & 1 & \dots & 1 \\
            0 & 0 & 1 & 1 & \dots & 1 \\
            0 & 0 & 0 & 1 & \dots & 1 \\
            \vdots & \vdots & \vdots & \vdots & \ddots & \vdots \\
            0 & 0 & 0 & 0 & \dots & 1
       \end{array} \right].
\]
Note that in the ``typical'' case that $\J$ contains no $\bot$'s, the number of $1$'s in $\free{\J}$ is exactly $d(d+1)/2$.  Given a \masking matrix, we write $\overline{\free{\J}}$ for its bitwise complement; i.e., $\overline{\free{\J}}^i_j = 1 - \free{\J}^i_j$.  We are now able to explain the notation used in the key step of the proof of the Main Theorem.  Given the semi-random matrix of \weights $\bWts$, note that for any $\J$,
\[
  \bWts = \bWJbar + \bWJ.
\]
Further, the random matrices $\bWJbar$ and $\bWJ$ are independent of one another.  This gives our crucial means of separating the random draw of $\bWts$ into two stages.

\subsection{The $\Trans$ mapping}

Here is the deterministic algorithm computing the $\Trans$ mapping:\\

\fbox{\parbox{6.0in}{
  $\Trans(x,\Wts) \co$
\texttt{
\begin{enumerate} \itemsep 0pt
  \item Set $\subitems_{d+1} = \allitems$.
  \item Initialize $\J$ to the length-$d$ column vector $(\bot, \bot, \dots, \bot)$.
  \item Initialize $\Y$ to the $n \times d$ matrix where every entry is $\bot$.
  \item For $t = d, d-1, d-2, \dots, 1$:
  \item \quad Let $\candidates_{t} = \{z \in \subitems_{t+1} : \Wts^{1..t} z > \Wts^{1..t} x\}$. \label{cands-def}
  \item \quad If $\candidates_t \neq \emptyset$,
  \item \quad \quad Set column $Y_t$ to be the $y \in \candidates_t$ for which $\Wts^{t+1} y$ is maximal.$^\dag$ \label{Y-def}
  \item \quad \quad Set $\J_t$ to be the least index in $[n]$ such that $\Y_t^{\J_t} \neq x^{\J_t}$.$^\ddag$
  \item \quad \quad Set $\subitems_t = \{z \in \subitems_{t+1} : \Wts^{t+1} z > \Wts^{t+1} \Y_t\ \text{and} \ z^{\J_t} = x^{\J_t}\}$. \label{subitems-def}
  \item \quad Else
  \item \quad \quad Set $\subitems_{t} = \subitems_{t+1}$.
  \item \quad End If
  \item End For
  \item Define the $n \times d$ matrix $\A$ by $\A^j_u = \begin{cases} Y^j_u & \text{if $j$ appears in $\J$,} \\ 0 & \text{otherwise.} \end{cases}$
  \item Define the \Boxx list $\BoxL = (\Bx_1, \dots, \Bx_d)$ as follows:
  \item[] For $u \in [d]$, if $\J_u = \bot$ then set $\Bx_u = \bot$.
  \item[] Otherwise, set $\Bx_u$ to be the \tbox{u} containing $\Wts x - (\free{\J} \circ \Wts) \A_u$.
  \item Output $(\J, \A, \BoxL)$.
\end{enumerate}}
}}

\medskip

\noindent {\footnotesize $^\dag$\ Two comments about this line:  Regarding maximality, say that we break ties by lexicographic order.

\noindent \phantom{$^\dag$\ }Regarding the case $t = d$, recall our abuse of notation: $\Wts^{d+1} y$ is defined to be $\objvec{y}{d+1}$.

\noindent $^\ddag$\ Such an index must exist: $\Y^t \neq x$ because $\Y^t \in \candidates_t$ and therefore $\Wts^{1..t} \Y^t > \Wts^{1..t} x$.}\\

\ignore{
\textbf{XXX Take a shot at explaining this algorithm in words. Tell the reader to forget about the ``else'' clause, for god's sake. XXX}\\}

It is clear that the index vector $\J$ and the \Boxx list $\BoxL$ output by $\Trans$ have the form we claimed.  We now verify that $\Trans(x,\Wts)$ indeed outputs a proper diagonalization matrix $\A$:
\begin{proposition} \label{prop:diagonalizes} The matrix $\A$ output by $\Trans(x,\Wts)$ always diagonalizes $x$ on $\J$.
\end{proposition}
\begin{proof}
  At the end of the algorithm, by definition row $\A^j$ is all zeros if $j$ does not appear in $\J$.  Thus it remains to analyze each row $\A^{\J_u}$, where $u \in [d]$ is such that $\J_u \neq \bot$. By definition, we have $\A^{\J_u}_t = \Y^{\J_u}_t$ for each $t \in [d]$.  Thus we need to verify the three conditions in~\eqref{eqn:diag-conds} for $\Y^{\J_u}_t$. First, $\Y^{\J_u}_t = \bot$ if and only if $\Y_t$ was ``not defined'' during iteration $t$ of the algorithm (i.e., if $\candidates_t = \emptyset$), which occurs precisely when $\J_t = \bot$.  Next, $\Y^{\J_u}_{u} = \overline{x^j}$ by definition of $\J_{u}$. Finally, because of line~\eqref{subitems-def} in $\Trans$ we have that $z^{\J_u} = x^{\J_u}$ for $z \in \subitems_u$.  Thus for any $t < u$ where $\J_t \neq \bot$ (and thus $\candidates_t \neq \emptyset$), we have $\Y^{\J_u}_t = x^{\J_u}$ because $\Y^{J_u} \in \candidates_t \subseteq \subitems_{t+1} \subseteq \subitems_u$.
\end{proof}

We also record another simple observation:
\begin{proposition} \label{prop:same-same}
  Given an execution of $\Trans(x,\Wts)$, any two \itms in $\subitems_t$ have the same $\J_t$'th coordinate, the same $\J_{t+1}$'th coordinate, \dots, and the same $\J_d$'th coordinate (excluding the cases $t \leq u \leq d$ where $\J_u = \bot$).
\end{proposition}
\begin{proof}
  For a fixed $t$ with $\J_t \neq \bot$, the fact that all \itms in $\subitems_t$ have the same $\J_t$'th coordinate follows immediately from the definition of $\subitems_t$.  The claim for coordinates $\J_{t+1}, \dots, \J_d$ follows from the fact that $\subitems_t \subseteq \subitems_{t+1} \subseteq \cdots \subseteq \subitems_{d}$.
\end{proof}

This proposition combines with our definition of \masking matrices in a crucial way:
\begin{named}{\Masking Lemma}  Given an execution of $\Trans(x, \Wts)$, for any $t \in [d]$ and $z, z' \in \subitems_t$,
\[
\Wts^{t} z > \Wts^{t} z' \quad\Leftrightarrow\quad (\WJbar)^{t} z > (\WJbar)^{t} z'.
\]
\end{named}
\begin{proof}
We have
\[
\Wts^{t} (z - z') = (\WJbar)^{t} (z - z') + (\WJ)^{t} (z - z').
\]
By definition of $\free{\J}$, the row vector $(\WJ)^{t}$ has nonzero entries only in indices $\J_{t}$, $\J_{t+1}$,~\dots,~$\J_{d}$ (excluding those $\J_u$'s which are~$\bot$).  But by Proposition~\ref{prop:same-same}, $z$ and $z'$ agree on these indices. Hence $(\WJ)^{t} (z - z') = 0$, and therefore $\Wts^{t} (z - z') = (\WJbar)^{t} (z - z')$.  The lemma follows.
\end{proof}

Finally, our proof of the key Uniqueness Lemma in Section~\ref{sec:uniqueness} will rely on the following simpler uniqueness claim:
\begin{named}{\Transcript--Determines--PO Lemma}  Suppose that we run $\Trans(x, \Wts)$, where $\Wts$ is an outcome for the \weights such that $x$ is Pareto optimal and such that $\OK$ occurs. Then at the end of the run, $x$ is uniquely defined by being the $z \in \subitems_1$ with maximal $\Wts^1 z$.
\end{named}
We remark that the assumption that $\OK(\Wts)$ occurs is stronger than necessary; we only need that $\Wts^i y \neq \Wts^i y'$ for all $i \in [d]$ and distinct $y, y' \in \allitems$ (an event that occurs with probability~$1$).
\begin{proof}
We make the following two claims about the execution of $\Trans(x, \Wts)$:

Claim 1:  For each $t \in [d+1]$ it holds that $x$ is not $t$-dominated by any $z \in \subitems_t$.

Claim 2: $x$ must be in $\subitems_1$.

\noindent Assuming these claims, the lemma follows immediately:  $x \in \subitems_1$ by claim~$2$, and no $z \in \subitems_1$ has $\Wts^{1} z \geq \Wts^{1} x$ by claim~$1$.

We begin by proving Claim~$1$.  For $t = d+1$, this follows immediately from the definition of $x$ being Pareto optimal.  For smaller $t$, let us consider the
$t$'th iteration of ``For'' loop, in which $\subitems_t$ is defined. We need to consider two cases corresponding to the ``If'' condition.  First suppose $\candidates_t \neq \emptyset$, so lines~\eqref{Y-def}---\eqref{subitems-def} are executed.  Now if there were some $z$ in the newly defined $\subitems_{t}$ which $t$-dominated $x$, then it would satisfy $\Wts^{t+1} z > \Wts^{t+1} \Y^t$ and $\Wts^{1..t} z \geq \Wts^{1..t} x$.  Since the $\OK$ event holds, the latter inequality can be strengthened to $\Wts^{1..t} z > \Wts^{1..t} x$.  But this means $z$ must be in the set $\candidates_t$.  Since $\Wts^{t+1} z > \Wts^{t+1} \Y^t$, we have a contradiction with how $\Y^t$ was chosen in line~\eqref{Y-def}.  We now consider the second case, that $\candidates_t = \emptyset$.  In this case, $\subitems_t = \subitems_{t+1}$.  Now by definition of $\candidates_t = \emptyset$, there is no $z \in \subitems_{t+1} = \subitems_t$ which has $\Wts^{1..t} z > \Wts^{1..t} x$.  Since the $\OK$ event occurs, we can strengthen this statement to say that no $z \in \subitems_t$ can even have $\Wts^{1..t} z \geq \Wts^{1..t} x$, as needed.

We now prove Claim~$2$.  Specifically, we show that $x \in \subitems_t$ for all $t \in [d+1]$ by (downward) induction on $t$.  The base case, that $x \in \subitems_{t+1}$, hold by definition. Assume then that $x \in \subitems_{t+1}$ for some $t \in [d]$.  Consider now the $t$'th iteration of the ``For'' loop.  If the ``If'' condition does not hold then $\subitems_t = \subitems_{t+1} \ni x$, as needed.  Assume then that lines~\eqref{Y-def}---\eqref{subitems-def} are executed.  To show $x \in \subitems_t$ it suffices to show that $\Wts^{t+1} x > \Wts^{t+1} \Y_t$.  If this is not true, then $\Wts^{t+1} \Y_t \geq \Wts^{t+1} x$, and $\Wts^{1..t} \Y_t \geq \Wts^{1..t} x$ also, since $\Y_t \in \candidates_t$.  But that means that $\Y_t \in \subitems_{t+1}$ $(t+1)$-dominates $x$, contradicting Claim~$1$.
\end{proof}

\section{The Uniqueness Lemma} \label{sec:uniqueness}

Let us restate the Uniqueness Lemma.
\begin{named}{Uniqueness Lemma}
Fix a \transcript $(\J,\A,\BoxL)$ and fix the outcome of $\bWJbar$.  Then there exists at most one $x \in \allitems$ such that the event
\begin{equation} \label{eqn:uniq-event}
\bigl(x \in \PO\bigr) \wedge \bigl(\Trans(x,\bWts) = (\J,\A,\BoxL)\bigr) \wedge \OK
\end{equation}
can occur.  Here the event's randomness is just the draw of $\bWJ$.
\end{named}

We prove the Uniqueness Lemma in a roundabout way.  Specifically, we introduce a second deterministic algorithm called $\Recon$, which takes as input a \transcript $(\J,\A,\BoxL)$ and fixed values $\WJbar$, and outputs \anitem $\ul{x} \in \allitems$ (or possibly `FAIL').  Instead of the Uniqueness Lemma as stated, we prove the following:

\begin{named}{Uniqueness Lemma$\boldsymbol{'}$}
Let \itm $x \in \allitems$ and \weight matrix $\Wts$ be such that $x$ is Pareto optimal and such that event $\OK$ occurs.  Assume further that $\Trans(x,\Wts) = (\J, \A, \BoxL)$.  Then \linebreak $\Recon((\J, \A, \BoxL), (\WJbar))$ outputs $x$.
\end{named}

This immediately implies the Uniqueness Lemma, as follows:  Fix a \transcript $(\J, \A, \BoxL)$ and an outcome $\bWJbar = \WJbar$.  Suppose there exist \itms $x, x' \in \allitems$ for which event~\eqref{eqn:uniq-event} can occur (with possibly different outcomes for $\bWJ$).  Then Uniqueness~Lemma$'$ tells us that the output of $\Recon((\J, \A, \BoxL), (\WJbar))$ is both $x$ and $x'$; hence $x = x'$.

The remainder of this section is devoted to the proof of Uniqueness~Lemma$'$.  We begin by defining the algorithm $\Recon$.

\fbox{\parbox{6.0in}{
  $\Recon((\J, \A, \BoxL), (\WJbar)) \co$
\texttt{
\begin{enumerate} \itemsep 0pt
  \item Set $\ul{\subitems_{d+1}} = \allitems$.
  \item Initialize $\ul{\Y}$ to the $n \times d$ matrix where every entry is $\bot$.
  \item For $t = d, d-1, d-2, \dots, 1$:
  \item \quad If $\J_t \neq \bot$,
  \item \quad \quad Write $b \in (\eps \Z)^t$ for the base point of $\Bx_t$. \label{eqn:write-b}
  \item \quad \quad Set $\ul{\candidates_t'} = \{z \in \ul{\subitems_{t+1}} : (\WJbar)^{1..t} z > b\ \text{and}\ z^j = A^j_t\ \forall\;\text{indices } j \in \J\}$.
  \item \quad \quad Set $\ul{\Y_t}$ to be the $y \in \ul{\candidates_t'}$ for which $(\WJbar)^{t+1} y$ is maximal.$^\ast$
  \item \quad \quad Set $\ul{\subitems_t} = \{z \in \ul{\subitems_{t+1}} : (\WJbar)^{t+1} z > (\WJbar)^{t+1} \ul{\Y_t}\ \text{and}\ z^{\J_t} \neq \ul{\Y_t}^{\J_t}\}$. \label{eqn:recon-subs}
  \item \quad Else
  \item \quad \quad Set $\ul{\subitems_{t}} = \ul{\subitems_{t+1}}$.
  \item \quad End If
  \item End For
  \item Output the $\ul{x} \in \ul{\subitems_1}$ for which $(\WJbar)^{1} \ul{x}$ is maximal. \label{last-line}
\end{enumerate}}
}}

\medskip

\noindent {\footnotesize $^\ast$\ Some comments about this line.  First, if $u = d$ then we interpret $(\WJbar)^{d+1} y$ to mean $\objvec{y}{d+1}$.  Second, regarding maximality, we break ties by lexicographic order.  Third, for some inputs to $\Recon$ it is possible that the set $\ul{\candidates_t'}$ is empty; in this case one can think of $\Recon$ as halting and outputting `FAIL'. However we will only be analyzing $\Recon$ on inputs where this provably never happen.  Finally, the first remark here also applies to line~\eqref{eqn:recon-subs} and the second and third remarks here also apply to line~\eqref{last-line}.}

\bigskip

To prove Uniqueness~Lemma$'$, we fix $x$ and $\Wts$ such that $x$ is Pareto optimal and such that event $\OK$ occurs.  We further suppose we have executed $\Trans(x,\Wts)$ producing $(\J, \A, \BoxL)$, and then executed $\Recon((\J, \A, \BoxL), (\WJbar))$ producing $\ul{x}$.  Our goal is to show that $\ul{x} = x$.  To do this, we will analyze the internal variable assignments made in the executions of $\Trans$ and $\Recon$.  More specifically, the main task will be to show the following claim asserting that $\ul{\subitems_t} = \subitems_t$ for all $t \in [d+1]$.  Once we show this, it will be easy to conclude that $\ul{x} = x$ also.

\begin{claim} $\ul{\subitems_t} = \subitems_t$ for all $t \in [d+1]$. \label{claim:unique}
\end{claim}
\begin{proof} The proof is by (downward) induction.  The base case is clear, as $\ul{\subitems_{d+1}} = \subitems_{d} = \allitems$. For the induction we assume that $\ul{\subitems_{u+1}} = \subitems_{u+1}$ for some $u \in [d]$. We now show that $\ul{\subitems_{u}} = \subitems_u$.  In doing so, we will also show that $\ul{\Y_u} = \Y_u$.  The set $\ul{\candidates_u'}$ will not necessarily equal $\ul{\candidates_u}$, but will be a subset of it.

We henceforth restrict attention to the the $t = u$ iteration of the ``For'' loop in the execution of $\Trans$ and $\Recon$, since this is when variables $\subitems_u$ and $\ul{\subitems_u}$ were set.  We consider two cases depending on whether or not $\J_u = \bot$.  In the easy case that $\J_u = \bot$, $\Trans$ must have assigned $\subitems_u = \subitems_{u+1}$, and certainly $\Recon$ assigned $\ul{\subitems_u} = \ul{\subitems_{u+1}}$.  By induction, $\ul{\subitems_{u+1}} = \subitems_{u+1}$, and hence $\ul{\subitems_u} = \subitems_u$ as required.

The remainder of the claim's proof is devoted to the case that $\J_u \neq \bot$, in which case $\Trans$ executed its lines~\eqref{Y-def}---\eqref{subitems-def} and $\Recon$ executed its lines~\eqref{eqn:write-b}---\eqref{eqn:recon-subs}. The $\Bx_u$ referred to in $\Recon$'s line~\eqref{eqn:write-b} is defined at the end of $\Trans$ to be \tbox{u} containing $\Wts x - (\WJ)\A_u$.  By definition, this means the base point $b \in (\eps \Z)^u$ used by $\Recon$ is such that
\newcommand{\newb}{{\widehat{b}}}
\begin{align}
 \Wts^{1..u} x - (\WJ)^{1..u} \A_u \ &\in\  b + [0,\eps)^u \nonumber\\
\Rightarrow  \qquad \Wts^{1..u} x \ &\in\  \newb + [0,\eps)^u, \nonumber\\
&\text{where } \newb = (\WJ)^{1..u} \A_u + b. \nonumber
\end{align}
Recall that $\Trans$ defines $\candidates_u = \{z \in \subitems_{u+1} \colon \Wts^{1..u} z > \Wts^{1..u} x\}$.  In fact, because we have assumed $\Wts$ causes event $\OK$ to occur, we may deduce
\begin{equation}
\candidates_u = \{z \in \subitems_{u+1} \colon \Wts^{1..u} z > \newb\}. \label{eqn:ok-dom}
\end{equation}
For if there were some $z \in \subitems_{u+1}$ and $i \in [u]$ with $\newb^i < \Wts^{i} z \leq \newb^i + \eps$, we would have $|\Wts^i z - \Wts^i x| \leq \eps$, contradicting the occurrence of $\OK$. (The reader may note that this deduction is precisely the reason we introduced the event $\OK$.)

Next, recall that $\Trans$ defines $\Y_u$ to be the $y \in \candidates_u$ for which $\Wts^{u+1} y$ is maximal (and this maximizer is unique since we assume $\OK$ occurs).  Since $\Y^j_u = \A^j_u$ for all indices $j$ appearing in $\J$, we must also have that $\Y_u$ is the maximizer of $\Wts^{u+1} y$ among all~$y$ within the following (nonempty) subset of $\candidates_u$:
\begin{equation} \label{eqn:candu1}
\candidates_u' \coloneqq \{z \in \subitems_{u+1} \colon \Wts^{1..u} z > \newb \ \text{and} \ z^j = \A_u^j \text{ for all indices $j \in \J$}\}.
\end{equation}
Observe that
\[
\Wts^{1..u} z > \newb  \quad\Leftrightarrow\quad (\WJbar)^{1..u} z + (\WJ)^{1..u} z > (\WJ)^{1..u} \A_u + b.
\]
Since all $z \in \candidates_u'$ agree with $\A_u$ in the indices from $\J$, and since $\free{\J}$ is nonzero only in columns whose indices are in $\J$, we have that
\begin{equation} \label{eqn:same}
(\WJ) z = (\WJ) \A_u \quad \text{for every $z \in \candidates_u'$.}
\end{equation}
Therefore an equivalent definition to~\eqref{eqn:candu1} is
\[
\candidates_u' = \{z \in \subitems_{u+1} \colon (\WJbar)^{1..u} z > b \ \text{and} \ z^j = \A_u^j \text{ for all indices $j \in \J$}\}.
\]
But $\subitems_{u+1} = \ul{\subitems_{u+1}}$ by induction, and hence $\candidates_u' = \ul{\candidates_u'}$.

The remainder of the proof of the claim now follows fairly easily using the \Masking~Lemma from Section~\ref{sec:trans}.  Recall that $\Y_u$ is the maximizer of $\Wts^{u+1} y$ among all $y \in \candidates_u'$.  On the other hand, $\Recon$ defines $\ul{\Y_u}$ to be the $y \in \ul{\candidates_u'} = \candidates_u'$ with maximal $(\WJbar)^{u+1} y$.  We claim that $\ul{\Y_u} = \Y_u$.  If $u = d$ then this is immediate, as both $\Wts^{d+1} y$ and $(\WJbar)^{d+1} y$ are interpreted as $\objvec{y}{d+1}$.  If $u < d$, this follows immediately from the \Masking~Lemma, using the fact that $\candidates_u' \subseteq \subitems_{u+1}$.

Finally, we wish to show that $\ul{\subitems_u} = \subitems_u$. Recall that
\begin{align*}
\subitems_u &= \{z \in \subitems_{u+1} \colon \Wts^{u+1} z > \Wts^{u+1} \Y_u \text{ and } z^{\J_u} = x^{\J_u}\}, \\
\text{and} \qquad \ul{\subitems_u} &= \{z \in \ul{\subitems_{u+1}} : (\WJbar)^{u+1} z > (\WJbar)^{u+1} \ul{\Y_u}\ \text{and}\ z^{\J_u} \neq \ul{\Y_u}^{\J_u}\} \\
&= \{z \in \subitems_{u+1} : (\WJbar)^{u+1} z > (\WJbar)^{u+1} \Y_u\ \text{and}\ z^{\J_u} = x^{\J_u}\};
\end{align*}
in this last deduction we used $\ul{\subitems_{u+1}} = \subitems_{u+1}$ (by induction), $\ul{\Y_u} = \Y_u$, and $\overline{\Y_u^{\J_u}} = x^{\J_u}$ (which follows from the definition of $\J_u$ in $\Trans$).  If $u = d$ then $\ul{\subitems_u} = \subitems_u$ again follows from the interpretation $\Wts^{d+1} z = (\WJbar)^{d+1} z  = \objvec{z}{d+1}$.  If $u < d$ then $\ul{\subitems_u} = \subitems_u$ again follows from the \Masking~Lemma, noting that $z, \Y_u \in \subitems_{u+1}$.  This completes the proof of the induction and hence the claim.
\end{proof}

Having proven Claim~\ref{claim:unique}, it is easy to complete the proof of Uniqueness~Lemma$'$, i.e., to show $\ul{x} = x$.  Since the \weights matrix $\Wts$ is assumed to make $x$ Pareto optimal and make $\OK$ occur, the \Transcript--Determines--PO Lemma from Section~\ref{sec:trans} implies that $x$ is the \itm $z \in \subitems_1$ with maximal $\Wts^1 z$.  On the other hand, $\ul{x}$ is defined to be the \itm $z \in \ul{\subitems_1} = \subitems_1$  with maximal $(\WJbar)^1 z$.  But these maximizers are equal by the \Masking~Lemma.

\section{The Boundedness Lemma} \label{sec:boundedness}

In this section we restate and prove the Boundedness Lemma.
\begin{named}{Boundedness Lemma} For every fixed $(\J,\A,\BoxL)$, outcome of $\bWJbar$, and $x \in \allitems$, it holds that
\[
\Pr_{\bWJ}[\Trans(x,\bWts) = (\J,\A,\BoxL)] \leq \pdfbd^{\dim(\BoxL)} \eps^{\dim(\BoxL)}.
\]
\end{named}
\begin{proof}
As in the proof of the Uniqueness Lemma we fix the \transcript $(\J,\A,\BoxL)$ and the outcome $\bWJbar = \WJbar$.  Unlike the proof of that lemma, we also fix $x \in \allitems$.  By Proposition~\ref{prop:diagonalizes} we may assume that matrix $\A$ diagonalizes $x$ on $\J$; otherwise the probability of $\Trans(\bWJ) = (\J,\A,\BoxL)$ is~$0$.

\newcommand{\IN}{\mathbf{IN}}
Write $\BoxL = (\Bx_1, \dots, \Bx_d)$, where each $\Bx_t$ is either a \tbox{t} or is $\bot$ (if $\J_t = \bot$). For each $t \in [d]$ with $\J_t \neq \bot$ we define the event
\[
\IN_t = \text{``}\bWts x - (\bWJ) \A_t \in \Bx_t\text{''},
\]
where again, the randomness of these events is just the draw of $\bWJ$. We may complete the proof by showing
\begin{equation} \label{eqn:bdd-bd}
\Pr_{\bWJ}\left[ \bigwedge_{t \in [d] : \J_t \neq \bot} \IN_t \right] \leq \pdfbd^{\dim(\BoxL)} \eps^{\dim(\BoxL)}.
\end{equation}
Recall that
\[
\free{\J}^i_j = \begin{cases}
                    1 &  \text{if $j = \J_t \in [n]$ for some $t \in [d]$ and $i \leq t$,}\\
                    0 & \text{otherwise.}
                \end{cases}
\]
We will imagine drawing the random entries of $\bWJ$ in $d$ stages.  In the $t$'th stage we draw the $t$ entries $\bWJ^{1..t}_{\J_t}$, unless $\J_t = \bot$ in which case we ``skip'' the $t$'th stage.  By the independence of the entries, the following claim immediately implies~\eqref{eqn:bdd-bd}:\\

\noindent Claim:  Assume $t \in [d]$ has $\J_t \neq \bot$.  Suppose we have completed the first $t-1$ stages of drawing $\bWJ$.  Then whether the event $\IN_t$ occurs is determined in the $t$'th stage, and its probability is at most $\pdfbd^t \eps^t$.\\

To prove the claim we we write $b \in \R^t$ for the base point of $\Bx_t$ and observe that
\begin{align}
\IN_t \quad\Leftrightarrow\quad \bWts x - (\bWJ) \A_t \ &\in\  \Bx_t \nonumber\\
\Leftrightarrow\quad (\WJbar)^{1..t}x + (\bWJ)^{1..t}x - (\bWJ)^{1..t} \A_t \ &\in\  b + [0,\eps)^t \nonumber\\
\Leftrightarrow\quad (\bWJ)^{1..t}(x -\A_t) \ &\in\  \bigl(b - (\WJbar)^{1..t}x\bigr) + [0,\eps)^t. \label{eqn:hit-box}
\end{align}
Recalling the definition of $\free{\J}$ we see that for a fixed $i \in [t]$,
\begin{align}
(\bWJ)^{i}(x -\A_t) = &\sum_{i \leq u < t : \J_u \neq \bot} (\WJ)^i_{\J_u} (x - \A_t)^{\J_u} \label{eqn:W-less}\\
&+ \  (\bWJ)^i_{\J_t} (x - \A_t)^{\J_t} \label{eqn:W-equal} \\
&+ \  \sum_{u > t : \J_u \neq \bot} (\bWJ)^i_{\J_u} (x - \A_t)^{\J_u}. \label{eqn:W-greater}
\end{align}
Please note that in~\eqref{eqn:W-less} we have written $\WJ$ rather than $\bWJ$ because the entries $(\bWJ)^i_{\J_u}$ for $u < t$ have been fixed prior to the $t$'th stage.  The entries of $\bWJ$ appearing in~\eqref{eqn:W-equal} and~\eqref{eqn:W-greater}, however, are still to be drawn.

At this point it may seem as though the event $\IN_t$ as given in~\eqref{eqn:hit-box} depends not only on the entries $(\bWJ)^{1..t}_{\J_t}$ as stated in the claim, but also on the entries $(\bWJ)^{1..t}_{\J_u}$ for $u > t$.  But this is where we make a crucial  observation; indeed, the one which explains why we defined $\Trans$ to produce diagonalization matrices.  By definition of $\A$ diagonalizing $x$ on $\J$,
\[
(x - \A_t)^j = \begin{cases}
                   \pm 1 & \text{if $j = \J_t$,} \\
                       0 & \text{if $j = \J_u \in [n]$ for some $u > t$.}
               \end{cases}
\]
(If $j = \J_u \in [n]$ for some $u < t$ then we cannot say anything about $(x - \A_t)^j$, but we do not need to.)\ignore{ {\textbf XXXdiagram here would be niceXXX.}} Substituting this into~\eqref{eqn:W-equal} and~\eqref{eqn:W-greater}, we deduce that
\begin{equation} \label{eqn:free-weights}
(\bWJ)^{i}(x -\A_t) = \text{constant } \pm (\bWJ)^i_{\J_t}.
\end{equation}
In particular, the term~\eqref{eqn:W-greater} has dropped out; hence event~\eqref{eqn:hit-box} does not in fact depend on the entries $(\bWJ)^{1..t}_{\J_u}$ for $u > t$, as claimed.  Finally, substituting~\eqref{eqn:free-weights} into~\eqref{eqn:hit-box} we see that the event $\IN_t$ is equivalent to a conjunction of $t$ events of the form
\[
\pm (\bWJ)^i_{\J_t} \in [c_i, c_i + \eps)
\]
where the $c_i$'s are fixed constants.  Since the random variables $(\bWJ)^i_{\J_t}$ are independent and have pdf's bounded by $\pdfbd$, we conclude that the probability of $\IN_t$ is indeed at most $(\pdfbd \eps)^t$, as claimed.
\end{proof}

\section{The Counting Lemma}
Here we restate and prove the Counting Lemma.
\begin{named}{Counting Lemma}  For a fixed $n$ and $\eps$, the quantity
\begin{equation} \label{eqn:target}
\sum_{\substack{\text{\emph{possible \transcripts{}}} \\ (\J, \A, \BoxL)}} \pdfbd^{\dim(\BoxL)} \eps^{\dim(\BoxL)}
\end{equation}
is at most $2 \cdot (4d\pdfbd)^{d(d+1)/2} \cdot n^{2d}$.
\end{named}
\begin{proof}
For a given index vector $\J$ let us define the following quantities:
\newcommand{\countt}{\mathrm{count}}
\newcommand{\summ}{\mathrm{sum}}
\[
\countt(\J) = \#\{t  : \J_t \neq \bot\}, \quad \summ(\J) = \sum \{t  : \J_t \neq \bot\}, \quad \max(\J) = \max \{t  : \J_t \neq \bot\}.
\]
Observe that for a possible \transcript $(\J, \A, \BoxL)$, the quantity $\summ(\J)$ is identical to $\dim(\BoxL)$.  We may therefore express~\eqref{eqn:target} as
\begin{equation} \label{eqn:final-sum}
\sum_{\text{possible } \J} \pdfbd^{\summ(\J)} \eps^{\summ(\J)} \cdot \#\{(\A, \BoxL) \text{ s.t.\ } (\J,\A,\BoxL) \text{ is a possible \transcript}\}.
\end{equation}
Let us now count the pairs $(\A, \BoxL)$ that form possible \transcripts with $\J$.  By Proposition~\ref{prop:diagonalizes} we know that $\A$ must diagonalize some \itm $x$ on $\J$. There are $2^{\countt(\J)}$ choices for the values of $x^{j}$, for $j$ appearing in $\J$. These force some entries of $\A$; the remaining $\sum \{ t-1 : \J_t \neq \bot\} = \summ(\J) - \countt(\J)$ entries are free.  Thus there are
\begin{equation} \label{eqn:A-count}
2^{\countt(\J)} 2^{\summ(\J) - \countt(\J)} = 2^{\summ(\J)} \ \text{possible choices for $\A$.}
\end{equation}

As for $\BoxL$, let us first count the number of possibilities for $\Bx_{\max(\J)}$ (assuming $\max(\J)$ exists).  We write $m = \max(\J)$ for brevity; on first reading, one should think of $m$ as always being $d$. An execution of $\Trans(x,\Wts)$ which is consistent with $\J$ and $\A$ defines $\Bx_{m}$ to be the \tbox{m} containing the point $p = \Wts x - (\free{\J} \circ \Wts) \A_m$.  Since the entries of $\Wts$ are bounded in $[-1,1]$ always and since $\free{\J}$ contains at most $d$ nonzero entries, the point $p$ must lie in $[-n-d, n+d)^d$.\footnote{Proving that $p$ cannot have any coordinate exactly equal to $n+d$ is an exercise for the reader.} There are therefore at most $(2(n+d)/\eps)^m$ choices for the \boxx $\Bx_m$.

We could similarly upper-bound the number of choices for each remaining \tbox{t} by $(2(n+d)/\eps)^t$; however, this would lead to a final count whose dependence on $d$ was $n^{d + d(d+1)/2}$, rather than $n^{2d}$. To get the much better dependence of $n^{2d}$ we observe that once $\Bx_m$ is chosen, the remaining \tboxes{t} cannot be ``too far away'' because, like $\Bx_m$, they contain a point close to $\Wts x$.  More precisely, let $t < m$ be such that $\J_t \neq \bot$ and consider $\Bx_t$.  It is the \tbox{t} containing $\hat{p} = \Wts x - (\free{\J} \circ \Wts) \A_t$.  Now $p - \hat{p} = (\free{\J} \circ \Wts)(\A_t - \A_u)$, which means that $\|p - \hat{p}\|_\infty \leq d$.  It follows that given the choice of $\Bx_m$, there are at most $((2d+1)/\eps)^t$ choices for $\Bx_t$.  We conclude that the number of possible choices for $\BoxL$ is at most
\begin{multline*}
(2(n+d)/\eps)^{\max(\J)} \cdot \prod_{t < \max(\J) : \J_t \neq \bot} ((2d+1)/\eps)^t = \left(\frac{2(n+d)}{2d+1}\right)^{\max(\J)} \cdot \left(\frac{2d+1}{\eps}\right)^{\summ(\J)} \\ \leq \left(\frac{2(n+d)}{2d+1}\right)^{d} \cdot \left(\frac{2d+1}{\eps}\right)^{\summ(\J)}.
\end{multline*}
Combining this with~\eqref{eqn:A-count} and substituting into~\eqref{eqn:final-sum}, we upper-bound~\eqref{eqn:target} by
\[
\sum_{\text{possible } \J} \bigl(2(2d+1)\pdfbd\bigr)^{\summ(\J)} \bigl(2(n+d)/(2d+1)\bigr)^d.
\]
Finally, we simply upper-bound $\summ(\J)$ by $d(d+1)/2$ and the number of possible $\J$ by $(n+1)^d$.  We conclude that~\eqref{eqn:target} is at most
\[
(n+1)^d \bigl(2(2d+1)\pdfbd\bigr)^{d(d+1)/2} \bigl(2(n+d)/(2d+1)\bigr)^d = (4\pdfbd)^{d(d+1)/2} (d+1/2)^{d(d-1)/2} (n+1)^d (n+d)^d.
\]
One may check that $(d+1/2)^{(d-1)/2} (n+1) (n+d)  \leq 2^{1/d} d^{(d+1)/2} n^2$ for any $d \geq 1$ and $n \geq 3$ (which we may assume, as our final bound is always at least $2^3$).  Hence~\eqref{eqn:target} is indeed at most
\[
2(4d \pdfbd)^{d(d+1)/2} n^{2d},
\]
as claimed.
\end{proof}

\section{Conclusion}
There are several open problems that remain.  One intriguing problem is to show a lower bound for the expected number of Pareto optima in which the exponent on $n$ grows with $d$.  Currently we cannot rule out the possibility of an upper bound of the form $f(d, \pdfbd) n^2$; however we regard this possibility as unlikely.  We feel it is likely that there is a lower bound of at least $\Omega(n^d)$ for constant $d$ and $\pdfbd$; our intuition is partly based on the known lower bound of $\Omega(n^d)$ in the scenario of $2^n$ completely independent points uniformly distributed on $[-1,1]^{d+1}$.  

Another interesting open problem is whether our methods can be used to give improved upper bounds on the higher moments of the number of Pareto optima in the smoothed analysis model.  This is currently unclear; we know of no bounds that improve on those of R\"oglin and Teng~\cite{RT09}.  Finally, one could ask about reducing the factor of $(\pdfbd d)^{d(d+1)/2}$ in our bound, as well as whether our results extend to the case of \itms in $\{0, 1, 2, \dots, c\}^n$ for integer constants $c >1$.

\subsection{Acknowledgements}

Part of this research was done during a visit to Microsoft Research New England; we thank them for their hospitality.  We also thank Shang-Hua Teng and Ilias Diakonikolas for sharing their expertise. 

\newpage

\bibliographystyle{alpha}
\bibliography{pareto}

\begin{thebibliography}{ANRV07}

\bibitem[ANRV07]{ANRV07}
Heiner Ackermann, Alantha Newman, Heiko R\"oglin, and Berthold V\"ocking.
\newblock Decision-making based on approximate and smoothed {P}areto curves.
\newblock {\em Theoretical Computer Science}, 378(3):253--270, 2007.

\bibitem[Bei04]{Bei04}
Ren\'e Beier.
\newblock {\em Probabilistic Analysis of Discrete Optimization Problems}.
\newblock PhD thesis, Universit\"at des Saarlandes, 2004.

\bibitem[BKS01]{BKS01}
Stephan B\"orzs\"ony, Donald Kossmann, and Konrad Stocker.
\newblock The {S}kyline operator.
\newblock In {\em Proceedings of the 17th Annual International Conference on
  Data Engineering}, pages 421--430, 2001.

\bibitem[BKST78]{BKST78}
Jon Bentley, Hsiang-Tsung Kung, Mario Schkolnick, and Clark Thompson.
\newblock Random knapsack in expected polynomial time.
\newblock {\em Journal of the ACM}, 25(4):536--543, 1978.

\bibitem[BRV07]{BRV07}
Ren\'e Beier, Heiko R\"oglin, and Berthold V\"ocking.
\newblock The smoothed number of {P}areto optimal solutions in bicriteria
  integer optimization.
\newblock In {\em Proceedings of the 11th Annual Conference on Integer
  Programming and Combinatorial Optimization}, pages 53--67, 2007.

\bibitem[Buc89]{Buc89}
Christian Buchta.
\newblock On the average number of maxima in a set of vectors.
\newblock {\em Information Processing Letters}, 33(2):63--65, 1989.

\bibitem[BV04]{BV04b}
Ren\'e Beier and Berthold V\"ocking.
\newblock Random knapsack in expected polynomial time.
\newblock {\em Journal of Computer and System Sciences}, 69(3):306--329, 2004.

\bibitem[BV06]{BV06}
Ren\'e Beier and Berthold V\"ocking.
\newblock Random knapsack in expected polynomial time.
\newblock {\em Typical Properties of Winners and Losers in Discrete
  Optimization}, 35(4):855--881, 2006.

\bibitem[Deb01]{Deb01}
Kalyanmoy Deb.
\newblock {\em Multi-objective optimization using evolutionary algorithms}.
\newblock Wiley, 2001.

\bibitem[Dev80]{Dev80}
Luc Devroye.
\newblock A note on finding convex hulls via maximal vectors.
\newblock {\em Information Processing Letters}, 11(1):53--56, 1980.

\bibitem[DF89]{DF89}
Martin Dyer and Alan Frieze.
\newblock Probabilistic analysis of the multidimensional knapsack problem.
\newblock {\em Mathematics of Operations Research}, 14(1):162--176, 1989.

\bibitem[Dia10]{Dia10}
Ilias Diakonikolas.
\newblock {\em Approximation of Multiobjective Optimization Problems}.
\newblock PhD thesis, Columbia University, 2010.

\bibitem[Ehr05]{Ehr05}
Matthias Ehrgott.
\newblock {\em Multicriteria optimization}.
\newblock Springer, 2005.

\bibitem[GMS84]{GM84b}
Andrew Goldberg and Alberto Marchetti-Spaccamela.
\newblock On finding the exact solution of a zero-one knapsack problem.
\newblock In {\em Proceedings of the 16th Annual ACM Symposium on Theory of
  Computing}, pages 359--368, 1984.

\bibitem[Lue98]{Lue98}
George Lueker.
\newblock Average-case analysis of off-line and on-line {K}napsack problems.
\newblock {\em Journal of Algorithms}, 29(2):277--305, 1998.

\bibitem[NU69]{NU69}
George Nemhauser and Zev Ullmann.
\newblock Discrete dynamic programming and capital allocation.
\newblock {\em Management Science}, 15(9):494--505, 1969.

\bibitem[PY02]{PY02}
Christos Papadimitriou and Mihalis Yannakakis.
\newblock On the approximability of trade-offs and optimal access of web
  sources.
\newblock In {\em Proceedings of the 43rd Annual IEEE Symposium on Foundations
  of Computer Science}, pages 86--92, 2002.

\bibitem[RT09]{RT09}
Heiko R\"oglin and Shang-Hua Teng.
\newblock Smoothed analysis of multiobjective optimization.
\newblock In {\em Proceedings of the 50th Annual IEEE Symposium on Foundations
  of Computer Science}, pages 681--690, 2009.

\bibitem[RV07]{RV07}
Heiko R\"oglin and Berthold V\"ocking.
\newblock Smoothed analysis of integer programming.
\newblock {\em Mathematical Programming}, 110(1):21--56, 2007.

\bibitem[ST04]{ST04}
Daniel Spielman and Shang-Hua Teng.
\newblock Smoothed analysis of algorithms: {W}hy the simplex algorithm usually
  takes polynomial time.
\newblock {\em Journal of the ACM}, 51(3):385--463, 2004.

\bibitem[Ten10]{Ten01}
Shang-Hua Teng, 2010.
\newblock National Science Foundation award \#0964481. Abstract available at
  \texttt{http://www.nsf.gov/awardsearch/showAward.do?AwardNumber=0964481}.

\end{thebibliography}

\end{document}